\documentclass[format=acmsmall, review=false]{acmart}
\usepackage{acm-null}
\usepackage{bbm,dsfont}
\usepackage{hyperref} 

\usepackage{booktabs} 
\usepackage[ruled]{algorithm2e} 

\SetAlFnt{\small}
\SetAlCapFnt{\small}
\SetAlCapNameFnt{\small}
\SetAlCapHSkip{0pt}
\IncMargin{-\parindent}

\newtheorem{theorem}{Theorem}
\newtheorem{lemma}[theorem]{Lemma}
\newtheorem{claim}[theorem]{Claim}

\newtheorem{corollary}[theorem]{Corollary}
\theoremstyle{definition}

\newtheorem{example}{Example}[]

\theoremstyle{remark}

\newcommand{\bigO}{\mathcal{O}}
\DeclareMathOperator*{\argmax}{\arg\!\max}
\DeclareMathOperator{\one}{\mathds{1}}

\DeclareMathOperator{\bin}{\mathcal{B}}

\DeclareMathOperator{\co}{co}

\def\x{\mathbf{x}}

\def\P{\mathbf{P}}

\newcommand{\Prr}[1]{\Pr\left[\,#1\,\right]}

\newcommand{\R}{\mathbb{R}}

\newcommand{\E}{\mathbb{E}}

\allowdisplaybreaks

\title{Impartial selection with prior information}

\author{Ioannis Caragiannis}
\orcid{0000-0002-4918-7131}
\affiliation{%
	\institution{Aarhus University}
	\department{Department of Computer Science}
	\streetaddress{{\AA}bogade 34}
	\city{Aarhus}
	\postcode{8240}
	\country{Denmark}
}
\email{iannis@cs.au.dk}

\author{George Christodoulou}
\affiliation{%
	\institution{University of Liverpool}
	\department{Department of Computer Science}
	\streetaddress{Ashton Building, Ashton Street}
	\city{Liverpool}
	\postcode{L69 3BX}
	\country{United Kingdom}
}
\email{G.Christodoulou@liverpool.ac.uk}

\author{Nicos Protopapas}
\affiliation{%
	\institution{University of Liverpool}
	\department{Department of Computer Science}
	\streetaddress{Ashton Building, Ashton Street}
	\city{Liverpool}
	\postcode{L69 3BX}
	\country{United Kingdom}
}
\email{G.Christodoulou@liverpool.ac.uk}

\begin{abstract}
  We study the problem of {\em impartial selection}, a topic that lies
  at the intersection of computational social choice and mechanism
  design. The goal is to select the most popular individual among a
  set of community members. The input can be modeled as a directed
  graph, where each node represents an individual, and a directed edge
  indicates nomination or approval of a community member to
  another. An {\em impartial mechanism} is robust to potential selfish
  behavior of the individuals and provides appropriate incentives to
  voters to report their true preferences by ensuring that the chance
  of a node to become a winner does not depend on its outgoing
  edges. The goal is to design impartial mechanisms that select a node
  with an in-degree that is as close as possible to the highest
  in-degree. We measure the efficiency of such a mechanism by the
  difference of these in-degrees, known as its {\em additive}
  approximation.

  Following the success in the design of auction and posted pricing
  mechanisms with good approximation guarantees for welfare and profit
  maximization, we study the extent to which prior information on
  voters' preferences could be useful in the design of efficient
  deterministic impartial selection mechanisms with good additive
  approximation guarantees. We consider three models of prior
  information, which we call the {\em opinion poll}, the {\em a priori
    popularity}, and the {\em uniform} model. We analyze the
  performance of a natural selection mechanism that we call {\em
    approval voting with default} (AVD) and show that it achieves a
  $\bigO(\sqrt{n\ln{n}})$ additive guarantee for opinion poll and a
  $\bigO(\ln^2n)$ for a priori popularity inputs, where $n$ is the
  number of individuals. We consider this polylogarithmic bound as our
  main technical contribution. We complement this last result by
  showing that our analysis is close to tight, showing an
  $\Omega(\ln{n})$ lower bound. This holds in the uniform model, which
  is the simplest among the three models.
\end{abstract}

\begin{document}

\maketitle	

\section{Introduction}\label{sec:intro}

We study the problem of {\em impartial selection}, which has recently
attracted a lot of attention from a social choice theory and mechanism
design point of view. The goal is to select the most popular
individual, among a set of community members. The input can be modeled as a directed graph, where each node represents
an individual, and a directed edge indicates nomination or approval of
a community member to another. A selection mechanism takes a graph as
input and returns a single node as the winner. This would be a trivial
task from the algorithmic point of view, but the challenge here is
that the true preferences of the individuals are private information
known only to them. In settings where each individual is
simultaneously a voter and a candidate and therefore has also personal
interest in becoming a winner, she may manipulate the mechanism and
misreport her true preferences if this could increase her chance to
win. An {\em impartial mechanism} is robust to such behavior, and
provides appropriate incentives to voters to report their true
preferences by ensuring that the chance of a node to become a winner,
does not depend on its outgoing edges.

Unfortunately, it is well-known that the obvious selection mechanism
that always returns the highest in-degree node as a winner, suffers
from possible manipulation, i.e., it is not impartial.  The challenge
is to design an impartial selection mechanism that selects a winner
with an in-degree that approximates well the highest in-degree.

Impartial selection, was introduced independently by Holzman and
Moulin~\cite{moulin2013} and Alon et al.~\cite{alon11}.  The former
work considered minimum axiomatic properties that impartial selection
rules should satisfy, while the latter quantified the efficiency loss
with the notion of the approximation ratio, defined as the worst-case
ratio of the maximum in-degree over the in-degree of the node
which is selected by the mechanism. This line of research 
concluded with the work of Fischer and Klimm~\cite{fischer2015}, who proposed randomized impartial
mechanisms with optimal approximation ratio.

It is well-known \cite{alon11,fischer2015} that the most challenging
nomination profiles in terms of their approximation ratio, for both
deterministic and randomized mechanisms, are those with small
in-degrees. In particular, all deterministic impartial mechanisms have
unbounded ratio, and this can be demonstrated even on inputs with two
nodes and of maximum in-degree of $1$ (see~\cite{alon11} for an
example). Inspired by this crucial observation, Bousquet et
al.~\cite{bousquet2014} designed impartial randomized mechanisms that
return a nearly optimal node when the maximum in-degree is high
enough. Caragiannis et al.~\cite{caragiannis2019impartial} went one
step further to quantify this effect, by advocating that {\em
  additive} approximation may be a more appropriate measure to evaluate
impartial mechanisms, and also by providing mechanisms with sublinear
additive approximation guarantees.

The most natural and well-studied selection rule is the {\em approval voting rule} (AV), which has received much attention in social choice theory~\cite{LS10}.
In our context, AV always returns the node with the highest
in-degree. Unfortunately, as already mentioned, this mechanism is not
impartial. The reason is that in case of a tie at the maximum degree,
some of the nodes involved in the tie may have incentive to vote
non-truthfully. Fortunately, there is a simple fix of this deficiency, which is inspired by the simpler {\em plurality with default} mechanism by Holzman and Moulin~\cite{moulin2013}; {\em in
  case of a tie, select as winner a predetermined/default
  node}. We refer to this modified version of AV as {\em approval voting with default} (AVD). Although (a careful implementation of) this tweak re-establishes impartiality, this
modification comes at a cost, as this preselection should be
independent of the input graph. Imagine a scenario, where there is a tie
between two nodes with the maximum degree. In the unfortunate situation
where the default node receives only a small number of votes, this might
lead to a poor additive approximation, linear in the number of nodes.

Most of the previous work
consider randomized mechanisms, hence the efficiency is measured in
expectation. However, in the design of selection mechanisms, 
determinism is arguably more desirable.  Unfortunately, all the
known deterministic mechanisms have very poor, linear additive
approximation, and it is wide open whether substantially better mechanisms exist.

In this work, we take a different route: we study the extent to which 
{\em prior information} on the preferences of the voters could allow the
design of deterministic impartial selection mechanisms with good additive
approximation guarantees.  
%
Our focus is on the analysis of AVD, for which our design choice boils
down to an effective choice of the default node, with the help of the
prior information. We assume that the preferences are drawn from a
probability distribution that is known to the mechanism. We assume
throughout {\em voter independence},\footnote{We should note that,
  with correlated distributions there is not much one can achieve
  (see Example~\ref{ex:1} in the appendix).} that is, the random
choice of preferences for each voter is independent of those of the
others.


We propose different models that capture several aspects of the
problem. In the {\em opinion poll} model, we assume that the prior
information concerns information about the preferences of different
(types of) voters. The designer has access to the probability
$p_i{(S)}$, with which voter $i$ (or all voters of type $i$) would
approve a subset $S$ of candidates.
The {\em a priori popularity}
model assumes that the designer has prior information about the
popularity of each candidate $j$, which is summarized by a scalar $p_j$.  
We assume that each candidate $j$ receives independently a vote from
each voter with probability $p_{j}$. As a special case, we also study the uniform model, in which every  candidate $j$ has the same popularity $p_j=p$.

Note that these models capture different information scenarios; the
former assumes that the designer has access to opinion poll statistics
for each (type of) voter, while the latter assumes that the designer
has access only to aggregate information about the popularity of a
candidate. This aggregation is over the whole population of voters, as
the actual information may be sanitized to preserve anonymity of those
(types) participated in the poll. Note that popularity may measure
other forms of biases over specific individuals. For example, consider
the situation in which a PC wants to decide the best paper award;
then, the a priori popularity of a paper could be a function of the
authors' esteem, affiliation, etc.

\subsection{Contribution and techniques}

	 
Our main focus is the analysis of the AVD mechanism. We begin with the
opinion poll model and, as a warm-up, in Section~\ref{sec:constant},
we present a simple mechanism that ignores the edges of the graph and
selects as winner a pre-selected node of {\em maximum expected
  in-degree}. We call this mechanism the \emph{constant mechanism},
and show that it is $\Theta(\sqrt{n\ln{n}})$-additive
(Theorem~\ref{UB:constant}, Theorem~\ref{LB:GeneralmodelLB}).
The AVD mechanism that selects as default the node of highest
expected in-degree, can only perform better than the constant
mechanism.
%
%
%
Our main and most technically involved result shows that this version
of the AVD is $\bigO(\ln^2n)$-additive in the a priori popularity
model (Theorem~\ref{thm:plurality-with-default}). We complement this result by showing that our analysis is
tight, up to a logarithmic factor: even for uniform inputs where all
candidates are a priori equally popular, there is a class of instances
for which AVD has additive approximation $\Omega(\ln{n})$, {\em for
  any} choice of the default node (Theorem \ref{thm:lb-uniform}).
	

The analysis of the constant mechanism serves multiple purposes. It
illustrates that when prior information is available, a low expected
additive approximation is achievable even by simple deterministic
mechanisms, and by the simplest statistic of the prior, that is the
expected in-degree. This is in sharp contrast to the no-prior case,
where deterministic mechanisms have a very poor performance for both
additive \cite{caragiannis2019impartial} and multiplicative
approximation \cite{alon11}.  Second, the analysis of the constant
mechanism is quite simple; e.g., the upper bound follows by a simple
application of the Hoeffding bound. However, it introduces some of the
techniques (such as tail inequalities and reverse Chernoff bounds)
that our strongest results in Sections~\ref{sec:apriori}
and~\ref{sec:lb-AVD} use. Finally, the upper bound on the expected
additive approximation of the constant mechanism serves as a benchmark
of efficiency for all impartial mechanisms with priors.

The analysis of AVD is considerably more involved than the
analysis of the constant mechanism. Roughly speaking, the important
quantity that affects the additive approximation is the difference
between the maximum in-degree and the in-degree of the default node
when two or more nodes are tied with the highest in-degree, times the
probability of this tie. In the a priori popularity model, the
in-degree $d$ of a node is a random variable following the binomial
probability distribution with parameters $n$ (the number of trials)
and $p$ (the probability that a trial is successful). Furthermore, the
in-degrees of different nodes are independent. Hence, bounding the
probability of a tie at the maximum in-degree is related to (but more
demanding than) bounding the probability that two out of many
independent binomial random variables take the same maximum value.

Unfortunately, even though problems of this kind have been studied in
the literature of applied probability and statistics (e.g.,
see~\cite{BSW94,ES96,ESS93}), the existing results have not been
proved useful for our purposes. When the difference between the
maximum in-degree and the in-degree of the default node is large,
Chernoff bounds can unsurprisingly be used to show that the
probability of a tie at maximum is negligible and, hence, the
contribution to the expectation of the quantity of interest is
negligible as well. The real challenge is when the difference of the
two in-degrees is small. In this regime, it turns out that we need
sharp bounds on the ratio $\Pr[d=x]/\Pr[d\geq x]$ (also called the
{\em hazard function}) for a binomial random variable $d$ and value
$x$ that is close to the expectation $\mu=pn$ of $d$ (i.e., so that
$\Pr[d\geq x]$ is only polynomially small in terms of $n$). As we show
in Lemma~\ref{lem:technical}, en route to proving
Theorem~\ref{thm:plurality-with-default}, the ratio
$\Pr[d=x]/\Pr[d\geq x]$ is at most
$\bigO\left(\sqrt{\frac{\ln{n}}{\min\{\mu,n-\mu\}}}\right)$ in this
case.

We believe that this technical tool can be of independent interest and
could find applications elsewhere. The bound is asymptotically tight;
its tightness for $p=1/2$ is exploited in the proof of our logarithmic
lower bound (Theorem~\ref{thm:lb-uniform}). The exact dependence on the quantity $\sqrt{\min\{\mu,n-\mu\}}$
is very important to achieve a polylogarithmic upper bound on the
additive approximation of AVD for every a priori popularity input. In
addition to the above crucial idea and Chernoff bounds, our proof in Theorem~\ref{thm:lb-uniform} involves an inverse
Chernoff bound to bound {\em from below} the probability that a
binomial random variable is far from its expectation. These statements
are less popular than Chernoff bounds but rather standard.




\subsection{Further related work}

Impartial selection was introduced independently by Alon et al.~\cite{alon11} and Holzman and Moulin~\cite{moulin2013}.
Alon et al.~\cite{alon11} proposed the {approximation ratio} as the fraction between the
highest in-degree and the (expected) in-degree of the
winner. They provided a simple, $4$-approximate randomized mechanism and noted that
no randomized impartial mechanism can achieve an approximation ratio less than $2$, even with randomization. If randomization is not allowed, however, the approximation ratio can be arbitrarily large.
Later on, Fischer and Klimm~\cite{fischer2015} introduced a $2$-approximation randomized mechanism, closing that gap.
Bousquet et al.~\cite{bousquet2014} proposed a randomized mechanism
with an arbitrarily close to optimal approximation ratio, provided that the maximum
in-degree of the graph is large enough.

Holzman and Moulin~\cite{moulin2013} considered various mechanisms
under the more restricted family of graphs where each node has
an out-degree equal to $1$. Among others, they proposed the plurality with
default mechanism, which can be seen as a version of AVD mechanism, tailored
to that family of inputs. 
They also came up with an important impossibility result regarding the
quality of impartial mechanisms: Any deterministic impartial mechanism can guarantee, either to never select $0$ in-degree nodes or to always select a unanimously nominated node, but never both. Variations of the problem are studied
in~\cite{bjelde2017,declipper2008,mackenzie2015,tamura2016characterizing,tamura2014impartial}.


Additive approximation was first studied by Caragiannis et
al.~\cite{caragiannis2019impartial}. Therein, they propose
simple randomized mechanisms with sub-linear additive approximation
guarantees. They also show that a specific class of deterministic
mechanisms cannot achieve additive approximation less that $n-1$ (i.e. the worst possible additive approximation), while
an equivalent class, in the randomized setting, cannot achieve additive approximation
better than $\Omega(\sqrt{n})$. For general deterministic mechanisms however,
they only show a lower bound of $2$, while no deterministic
mechanism is known with additive approximation smaller than
$n-1$. Closing this gap remains a tantalizing open question.

Impartiality is encountered in various domains. In the AI
literature, a related application is
peer-reviewing~\cite{aziz2019strategyproof,kahng2018ranking,Kurokawa2015,mattei2020peernomination}.
In another direction, Babichenko et
al.~\cite{babichenko2018incentive,babichenko2020incentive} present
impartial mechanisms for the selection of the most influential node in a
network. The main difference with our setting is that the influence of a specific
node does not depend merely on its in-degree, but also on all the paths leading to that node.
 Mackenzie~\cite{mackenzie2019axiomatic} analyses the papal conclave through the lens of impartiality.

Our motivation for considering prior information comes from its
successful application to auction and posted pricing mechanisms. An
excellent survey of related work can be found
in~\cite{hartline2013bayesian}. We should also note that there is an
interesting connection of the techniques needed for the analysis in
the a priori popularity model with the literature on random
graphs~\cite{bollobas2001random,frieze2016introduction} and, in
particular, results regarding the multiplicity of the highest
in-degree in $G_{n,p}$ graphs. Unfortunately, such results have a
focus on asymptotics: for example, en route to proving bounds on the
chromatic number, Erd{\H{o}}s and Wilson~\cite{erdos1977chromatic}
showed that the maximum degree is unique with probability $1-o(1)$ in
$G_{n,1/2}$ graphs. Instead, for proving our approximation guarantees,
we need sharp estimations of the hidden $o(1)$ term. So, such results
are not directly applicable to our analysis.
 
\subsection{Roadmap}
The rest of the paper is structured as follows. We begin with
preliminary definitions and tail inequality statements in
Section~\ref{sec:prelim}. Section~\ref{sec:constant} is devoted to the
analysis of the constant mechanism in the opinion poll model. Our
polylogarithmic additive approximation for AVD is presented in
Section~\ref{sec:apriori} and the logarithmic lower bound in
Section~\ref{sec:lb-AVD}. We conclude with open problems in
Section~\ref{sec:open}. Two additional observations are given in the
appendix.



	\section{Preliminaries}\label{sec:prelim}
We denote by $N$ the set of individuals (or {\em agents}). For a set $S\subset N$, we use $N_S$ as an abbreviation of $N\setminus S$ and write $N_{i,...,j}$ instead of $N_{\{i,...,j\}}$ for simplicity.  A \emph{nomination profile} $G=(N,E)$ is a directed graph without self-loops that has the agents of $N$ as nodes. Each directed edge $ (i,j) \in E$ represents a nomination from agent $i$ to agent $j$.  Occasionally, we refer to the outgoing edges as $\emph{votes}$. We define as $x_i=\{ (i,j) \in E \}$ the set of outgoing edges from node $i \in N$ and use the tuple $\x=(x_1,...,x_n)$ as an alternative representation for $G$. We use $\x_{-i}$ to denote the graph $( N,E \setminus ( \{i\} \times N ) )$.  

Denoting by $\mathcal{G}$ the set of all nomination profiles over the agents of $N$, a (deterministic) {\em selection mechanism} is simply a function $f:\mathcal{G}\rightarrow N$ which maps each nomination profile to a single  node (the {\em winner}). A deterministic selection mechanism $f$ is {\em impartial} when for any agent $i\in N$, any graph $\x \in \mathcal{G}$ and any set $x'_i$ of outgoing edges from node $i$, it is $f(\x)=i$ if and only if $f(x'_i,\x_{-i})=i$. In other words, no agent (node) has any incentive to misreport her preferences (its outgoing edges).

We use $d_{j}(S,\x)$ to denote the in-degree of node $j \in N$, taking into account only the incoming edges from nodes of set $S$, given the profile $\x$, i.e.,
\[d_j(S,\x) = |\{i \in S: (i,j)\in E\}|.\] 
We use the simplified notations $d_j(\x)$ when $S=N_j$. $\Delta(\x)$ denotes the maximum in-degree of the profile $\x$, i.e.,  $\Delta(\x) =  \max_{j \in N} d_j(\x)$. Following the work of Caragiannis et al.~\cite{caragiannis2019impartial}, we evaluate the performance of a mechanism $f$ on a nomination profile $\x$ using the additive approximation $\Delta(\x)-d_{f(\x)}(\x)$, i.e., the difference between the maximum in-degree over all nodes and the in-degree of the winner returned by mechanism $f$. 

We assume that the input is a random nomination profile (among the agents of $N$), selected according to a probability distribution $\mathbf{P}$ over all such profiles. We assume {\em voter independence}, which means that the distribution $\mathbf{P}$ is a product $\prod_{i\in N}{\mathbf{P}_i}$ of independent distributions, where $\mathbf{P}_i$ denotes the distribution according to which node $i$ selects its set of outgoing edges.

We examine a hierarchy of three families of distributions, giving raise to {\em opinion poll}, {\em a priori popularity}, and {\em uniform} instances (or models), respectively:
\begin{itemize}
\item In the opinion poll model, each node $i\in N$ selects its set of outgoing edges among all possible edges to nodes of $N_i$, according to the probability distribution $\mathbf{P}_i$. Due to voter independence, the in-degree $d_j(\x)$ of each node $j$ is equal to the sum $\sum_{i\in N_j}{x_{ij}}$ of independent Bernoulli random variables, each denoting whether the directed edge from node $i$ to node $j$ exists in the nomination profile $(x_{ij}=1$) of not ($x_{ij}=0$). For simplicity of exposition, in our proofs, we consider $N$ to have $n+1$ agents; then, the in-degree of each node is the sum of $n$ independent random variables.
\item The a priori popularity model is the special case of opinion poll where each node $j$ has a popularity $p_j\in [0,1]$ and the directed edge $(i,j)$ exists in the nomination profile with probability $p_j$, independently on all other edges. In this case, the in-degree of node $j$ follows the binomial distribution $\bin(n,p_j)$, where $n$ denotes the number of trials and $p_j$ is the success probability for each trial. 
\item We call uniform the special case of the a priori popularity model with $p_j=p$ for every agent~$j$.
\end{itemize}
We assume that {\em prior information} about the underlying probability distributions is known in advance. Hence, we examine selection mechanisms that are defined using this information and evaluate them in terms of their expected additive approximation 
\[\E_{\x\sim{\P}}[\Delta(\x)-d_{f(\x)}(\x)].\]

We use the term $\alpha$-additive to refer to a selection mechanism with expected additive approximation at most $\alpha$.
Our aim is to design deterministic impartial selection mechanisms that have as low as possible expected additive approximation in any distribution from the above classes. Our positive results apply to opinion poll or to a priori popularity distributions; our proofs of negative results use the simplest uniform ones.

\subsection{Tail inequalities}
We include some tail bounds here that will be very useful later in our analysis.

\begin{lemma}[Hoeffding \cite{H63}]\label{lem:hoeffding}
	Let $X_1, X_2, ..., X_n$ be independent random variables so that  $\Prr{a_j\leq X_j \leq b_j} =1$. Then, the expectation of the random variable $X=\sum_{j=1}^n{X_j}$ is $\mathbb{E}[X]=\sum_{j=1}^n{\mathbb{E}[X_j]}$ and, furthermore, for every $\nu\geq 0$, $$\Prr{|X - \mathbb{E}[X]| \geq \nu}\leq 2\exp\left(-\frac{2\nu^2}{\sum_{j=1}^n{(b_j-a_j)^2}}\right).$$
\end{lemma}

\begin{lemma}[Chernoff bounds]\label{lem:chernoff}
	Let $B\sim\bin(n,p)$ and $\mu=np$. Then, the following inequalities hold 
	\begin{itemize}
		\item Let $x\geq \mu$. Then 
		\begin{align}\label{eq:okamoto}
		\Pr[B\geq x] &\leq \exp\left(-\frac{(x-\mu)^2n}{2\mu(n-\mu)}\right)
		\end{align}
		if $\mu\geq n/2$, and
		\begin{align}\label{eq:chernoff-upper}
		\Pr[B\geq x] &\leq \exp\left(-\frac{(x-\mu)^2}{3\mu}\right)
		\end{align}
		if $\mu<n/2$ and, furthermore, $x\leq 2\mu$.
		\item Let $x\leq \mu$. Then, 
		\begin{align}\label{eq:okamoto-lower}
		\Pr[B\leq x] &\leq \exp\left(-\frac{(\mu-x)^2n}{2\mu(n-\mu)}\right)
		\end{align}
		if $\mu\leq n/2$, and
		\begin{align}\label{eq:chernoff-lower}
		\Pr[B\leq x] &\leq \exp\left(-\frac{(\mu-x)^2}{3(n-\mu)}\right)
		\end{align}
		if $\mu>n/2$ and, furthermore, $x\geq 2\mu-n$.
	\end{itemize}
\end{lemma}
Inequalities (\ref{eq:chernoff-upper}) and (\ref{eq:chernoff-lower}) are the standard Chernoff bounds; e.g., see \cite{MR95}. Inequalities (\ref{eq:okamoto}) and (\ref{eq:okamoto-lower}) are due to Okamoto~\cite{O58}. The following lemma (see \cite{Ash90}, Lemma~4.7.2, page 116]) indicates that Chernoff bounds are asymptotically tight. 
\begin{lemma}\label{lem:reverse-chernoff}
	Let $B\sim \bin(n,p)$ and $\delta\in [0,1-p)$. Then,
	\begin{align*}
	\Pr[B \geq n(p+\delta)] &\geq \frac{1}{\sqrt{8n(p+\delta)(1-p-\delta)}} \cdot \left( \left(\frac{p}{p+\delta}\right)^{p+\delta} \left(\frac{1-p}{1-p-\delta}\right)^{1-p-\delta}\right)^n.
	\end{align*}
\end{lemma}
In particular, we will utilize the following cleaner statement, which follows easily by Lemma~\ref{lem:reverse-chernoff}. 
\begin{corollary}[Inverse Chernoff bound]\label{cor:inverse}
	Let $B\sim\bin(n,1/2)$ and $\delta\in [0,1/10]$. Then,
	\begin{align*}
	\Pr\left[B\geq n\left(\frac{1}{2}+\delta\right)\right] &\geq \frac{1}{\sqrt{2n}}\exp\left(-3\delta^2n\right)
	\end{align*}
\end{corollary}

\begin{proof}
By applying Lemma~\ref{lem:reverse-chernoff} to the random variable $B$, we have
		\begin{align*}
		\Pr\left[B\geq n\left(\frac{1}{2}+\delta\right)\right] &\geq \frac{1}{\sqrt{2n}}\left(\left(\frac{1/2}{1/2+\delta}\right)^{1/2+\delta}\left(\frac{1/2}{1/2-\delta}\right)^{1/2-\delta}\right)^n\\ 
		&= \frac{1}{\sqrt{2n}}\left(\frac{1}{\sqrt{1-4\delta^2}}\left(\frac{1/2-\delta}{1/2+\delta}\right)^{\delta}\right)^n\\
		&\geq \frac{1}{\sqrt{2n}} \exp\left(\frac{-2\delta^2-4\delta^3}{1-2\delta}n\right)  \geq \frac{1}{\sqrt{2n}} \exp\left(-3\delta^2n\right).
		\end{align*}
		The first inequality follows by Lemma~\ref{lem:reverse-chernoff} and since $(p+\delta)(1-p-\delta)$ is at most $1/4$. The second inequality follows by the inequality $e^z\geq 1+z$ for $z\in \R$ which implies that $\sqrt{1-4\delta^2}\leq \exp(-2\delta^2)$ and $\frac{1/2+\delta}{1/2-\delta}\leq \exp\left(\frac{4\delta}{1-2\delta}\right)$. The third inequality follows since $\delta\leq 1/10$.
	\end{proof}



\section{Warming up: the constant mechanism}
\label{sec:constant}
We first consider a simple mechanism, which we call the {\em constant} mechanism. This mechanism ignores all edges and awards a particular preselected node, which we call the {\em default winner} (or default node). The selection of the default winner depends only on the prior. For example, the criterion that we consider here is to select as default winner a node of maximum expected in-degree, i.e., 
\begin{align*}
f_{\co} \in \argmax_{v\in N}{\E[d_v(\x)]}.
\end{align*}
Our first statement is an upper bound on the additive approximation of the constant mechanism; its proof 
follows by a simple application of the Hoeffding bound (Lemma~\ref{lem:hoeffding}).


\begin{theorem}\label{UB:constant}
	For opinion poll inputs, the constant mechanism that uses the maximum expected in-degree node as the default winner has expected additive approximation $\bigO\left(\sqrt{n \ln{n}}\right)$. 
\end{theorem}

\begin{proof} Recall that, in the opinion poll model, the in-degree of node $v$ is the sum of $n$ independent Bernoulli random variables, i.e., $d_v(\x)=\sum_{u\in N_{v}}{x_{uv}}$. Then, a simple application of the Hoeffding bound (Lemma~\ref{lem:hoeffding}) yields
	\begin{align*}
	\Pr\left[d_v(\x)\geq \E[d_v(\x)]+\sqrt{n\ln{n}}\right] & \leq  \Pr\left[|d_v(\x)- \E[d_v(\x)]|\geq \sqrt{n\ln{n}}\right] \leq \frac{2}{n^2}.
	\end{align*}
	Hence, the probability that some node has in-degree at least $\E[d_{f_{\co}}(\x)]+\sqrt{n\ln{n}}$ is at most the probability that some node $v$ has in-degree at least $\E[d_v(\x)]+\sqrt{n\ln{n}}$. By the inequality above and the union bound, this probability is at most $\frac{2}{n^2}\cdot(n+1)\leq \frac{3}{n}$. Thus, the expected maximum in-degree is 
	\begin{align*}
	\E[\Delta(\x)] 	& \leq \E[d_{f_{\co}}(\x)]+\sqrt{n\ln{n}}+n \cdot \frac{3}{n} \leq \E[d_{f_{\co}}(\x)]+3+\sqrt{n\ln{n}},
	\end{align*}
	and the expected additive approximation $\E[\Delta(\x)-d_{f_{\co}}(\x)]$ is no more than $3+\sqrt{n\ln{n}}$.
\end{proof}

The bound in Theorem~\ref{UB:constant} is asymptotically tight. The lower bound instances that we use in the proof of the next statement are the simplest ones: uniform instances with $p=1/2$. Consequently, it holds for any selection of the default winner. The proof exploits the reverse Chernoff bound (Corollary~\ref{cor:inverse}).
 
\begin{theorem}\label{LB:GeneralmodelLB}
	The constant mechanism has expected additive approximation $\Omega\left(\sqrt{n \ln{n}}\right)$, even when applied to uniform inputs.
\end{theorem}

\begin{proof}
	Consider a uniform prior with $p=1/2$ over $n+1$ nodes, where $n$ is large, e.g., $n\geq 80$. Then, the in-degree of any node $u$ is a random variable following the binomial probability distribution $\bin(n,1/2)$. Let $u^*=f_{\co}$ be the node returned by the constant mechanism; clearly, $\E[d_{u^*}(\x)]=n/2$. Denote by $\mathcal{E}$ the event that some node different than $u^*$ has in-degree at least $\frac{n}{2}+\sqrt{\frac{n\ln{n}}{6}}$. By applying Corollary~\ref{cor:inverse} with $\delta=\sqrt{\frac{\ln{n}}{6n}}$ (the fact that $n$ is large guarantees that $\delta\leq 1/10$) to the random variable $d_u(\x)$, we have  
	\begin{align*}
	\Pr\left[d_u(\x) \geq \frac{n}{2}+\sqrt{\frac{n\ln{n}}{6}}\right] &\geq \frac{1}{n\sqrt{2}},
	\end{align*}
	for every node $u\not=u^*$ and, hence,
	\begin{align*}
	\Pr[\mathcal{E}] &\geq 1-\left(1-\frac{1}{n\sqrt{2}}\right)^n \geq 1-e^{-1/\sqrt{2}}\geq \sqrt{2}-1,
	\end{align*}
	where the second inequality follows by the inequality $(1-r/n)^n\leq e^{-r}$ and the third one by the inequality $e^z\geq 1+z$ (and, thus, $e^{1/\sqrt{2}}\geq 1+1/\sqrt{2}$). We now have 
	\begin{align*}
	\E[\Delta(\x)] & \geq \E[\max_{u\not=u^*}{d_u(\x)}\one\{\mathcal{E}\}]+\E[d_{u^*}(\x) \one\{\overline{\mathcal{E}}\}] \geq \left(\frac{n}{2}+\sqrt{\frac{n\ln{n}}{6}}\right)\Pr[\mathcal{E}] +\E[d_{u^*}(\x)]\Pr[\overline{\mathcal{E}}]\\
	&= \E[d_{u^*}(\x)]+\sqrt{\frac{n\ln{n}}{6}} \cdot\Pr[\mathcal{E}]\geq \E[d_{u^*}(\x)]+\frac{1}{6}\sqrt{n\ln{n}},
	\end{align*}
	and the desired lower bound on the expected additive approximation $\E[\Delta(\x)-d_{u^*}(\x)]$ follows.
\end{proof}

\section{A priori popularity and the AVD mechanism}\label{sec:apriori}
We devote this section to AVD mechanism and its analysis on a priori popularity instances. AVD uses a preselected node $t$ as the default winner. To give a formal definition of the mechanism, we say that a non-default node $k$ {\em beats} another non-default node $j$ in the nomination profile $\x$ if $d_k(N_{j,k,t},\x)>d_j(N_{j,k,t},\x)$, i.e., if node $k$ has higher in-degree than node $j$ when ignoring incoming edges from nodes $j$, $k$, and the default node $t$. Node $k$ beats (respectively, is beaten by) the default node $t$ if $d_k(N_{k,t},\x)>d_t(N_{k,t},\x)$ (respectively, $d_k(N_{k,t},\x)<d_t(N_{k,t},\x)$). When applied on the nomination profile $\x$, AVD returns as the winner $w$ the node that beats every other node, or the default node if no node that beats every other node exists.\footnote{The case in which no node beats every other node refines the notion of a tie that we informally used in Section~\ref{sec:intro}.} We remark that the default node is not prohibited to win by beating every other node.

Notice that, by misreporting its outgoing edges, a node cannot affect the set of other nodes it beats. Hence, AVD is clearly impartial. In addition, the above formal definition allows us to observe that the in-degree of the winner returned by AVD is never lower than the in-degree of the default node $t$. Indeed, when the default node is not the winner, it is beaten by the winner, who has at least as high in-degree. Hence, the upper bound of $\bigO(\sqrt{n\ln{n}})$ on the expected additive approximation of the constant mechanism on opinion poll instances carries over to AVD mechanism when the default node is selected to be a node of highest expected in-degree. In the following, we present a much stronger result that applies specifically to a priori popularity instances.
  
\begin{theorem}\label{thm:plurality-with-default}
	The AVD mechanism that uses the node of highest expected in-degree as the default node has an expected additive approximation of $\bigO(\ln^2n)$ when applied on a priori popularity instances.
\end{theorem}

Again, for simplicity of notation, in our analysis of AVD, we consider profiles with $n+1$ nodes. We assume that the number of nodes is large, e.g., $n\geq 10^6$ (otherwise, Theorem~\ref{thm:plurality-with-default} holds trivially). Let $p_k$ be the popularity of node $k$ and recall that the in-degree $d_k(\x)$ of node $k$ is a random variable taking values between $0$ and $n$ following the binomial distribution $\bin(n,p_k)$, which is also independent on the in-degree of the other nodes. Also, let $\mu_k=\E[d_k(\x)]=p_kn$ and $\xi_k=\min\{\mu_k,n-\mu_k\}$. 

We first consider the case of $\xi_t<8200\ln{n}$. This means that the expected in-degree of the default node is either very high, i.e., $\mu_t>n-8200\ln{n}$, or very low, i.e., $\mu_t<8200\ln{n}$. When $\mu_t>n-8200\ln{n}$, the expected degree of the winner (be it the default node or not; recall the argument above that compares AVD with the constant mechanism) is more than $n-8200\ln{n}$. As $d_w\leq n$, the expected additive approximation is less than $8200\ln{n}$. In the case $\mu_t<8200\ln{n}$, a simple application of a Chernoff bound (i.e., the tail inequality (\ref{eq:chernoff-upper}) from Lemma~\ref{lem:chernoff}) yields that $\Pr[d_k(\x)\geq 9000\ln{n}] \leq n^{-26}$ and, hence, the expected additive approximation is at most $ n^{-26}\cdot n + (1-n^{-26})\cdot 9000\ln{n} \leq 1+ 9000\ln{n}$.

So, in the following, we analyze the AVD mechanism assuming that $\xi_t\geq 8200\ln{n}$. Let $h$ be the highest in-degree among all nodes. Denote by $A$ the event that there is no node that beats every other node and the default node is beaten by a non-default node of degree $h$. The following lemma addresses the simplest case where the event $A$ is not true.
\begin{lemma}\label{lem:d-star-minus-d_t-leq-1}
	$\E[(h-d_w(\x))\one\{\overline{A}\}] \leq 1$.
\end{lemma}

\begin{proof}
	If the event $A$ does not hold, there must either be a node that beats every other node or the default node is not beaten by any node of degree $h$.
	
	So, first, assume that there is a node $w$ that beats every other node. The lemma follows if this node has degree $h$. Otherwise, let $i$ be a node of degree $h$. Since $w$ beats $i$, we have $d_w(\x)\geq d_w(N_{i,w,t},\x) \geq d_i(N_{i,w,t},\x) +1 \geq d_i(\x)-1 = h-1$ if $w\not=t$, and $d_w(\x)\geq d_t(N_{i,t},\x) \geq d_i(N_{i,t},\x) +1 \geq d_i(\x) = h$ if $w=t$.
	
	Now, assume that the default node is not beaten by any node of in-degree $h$ and there is no node that beats every other node. In this case, the winner will be the default node $t$. The lemma clearly follows if $t$ has in-degree $h$. Otherwise, since $t$ is not beaten by some node $i$ of degree $h$, we have $d_t(\x)\geq d_t(N_{i,t},\x) \geq d_i(N_{i,t},\x) \geq d_i(\x)-1=h-1$.
\end{proof}

We will now bound $\E[(h-d_w(\x))\one\{A\}]$; to do so, we will use a structural lemma.

\begin{lemma}\label{lem:bound-apx-event-A}
	Assume that $A$ is true and let $i$ be a node of highest in-degree $h$ that beats the default node $t$. Then, there is a node $j$, different than $i$ and $t$, that has degree either $h$, or $h-1$, or $h-2$.
\end{lemma}

\begin{proof}
	Since node $i$ does not beat every other node, there must be some node $j$ that is not beaten by $i$ (clearly, $j$ is different than $t$). Then, $d_j(\x)\geq d_j(N_{i,j,t},\x)\geq d_i(N_{i,j,t},\x)\geq d_i(\x)-2=h-2$.
\end{proof}

By Lemma~\ref{lem:bound-apx-event-A}, we can bound $\E[(h-d_w(\x))\one\{A\}]$ by the expected value of the difference $h-d_t(\x)$ for all possible values of the maximum degree $h$, all possibilities for an agent $i\not=t$ having degree $h$ and an agent $j\not=i,t$ having degree either $h$ or $h-1$ or $h-2$, with the degree of the default node ranging from $0$ to $h$ and the degree of all other nodes ranging from $0$ to $h$ as well. We have
\begin{align}\nonumber
\E[(h-d_w(\x))\one\{A\}]
&\leq \sum_{h=0}^{n}{\sum_{g=0}^{h}{(h-g) \cdot \Pr[d_t(\x)= g]\sum_{i\in N_t}{\Pr[d_i(\x)=h]}}}\\\label{eq:THE-sum}
& \quad\quad\quad\quad {{\cdot \sum_{j\in N_{i,t}}{\Pr[\max\{0,h-2\} \leq d_j(\x)\leq h]}\prod_{k\in N_{i,j,t}}{\Pr[d_k(\x)\leq h]}}}
\end{align}

For every $k\in \{1,...,n+1\}$, define the {\em comfort zone} $Z_k$ of agent $k$ to be the set of integers $\{L_k, ..., U_k\}$ with the boundaries satisfying $L_k\leq \mu_k\leq U_k$ and being defined as follows. The lower boundary $L_k$ is equal to the highest integer $c$ such that  $\Pr[d_k(\x)<c]\leq n^{-5.33}$ or $0$ if no such $c$ exists. The upper boundary $U_k$ is equal to the lowest integer $c$ such that  $\Pr[d_k(\x)>c]\leq n^{-5.33}$ or $n$ if no such $c$ exists. We use the terms ``above $Z_k$'' and ``below $Z_k$'' to denote the ranges of integers (if any) $\{0, ..., L_k-1\}$ and $\{U_k+1, ..., n\}$.

Now, by simple properties of the binomial distribution and the fact that node $t$ has maximum expected in-degree, we observe that if $h$ lies above the comfort zone of agent $t$, it also lies above the comfort zone of agent $i$ and, hence, $\Pr[d_i(\x)=h] \leq n^{-5.33}$. Also, if $g$ lies below the comfort zone $Z_t$, it holds $\Pr[d_t(\x)=g] \leq n^{-5.33}$. Furthermore, if $h-2$ lies above the confort zone $Z_j$, then $\Pr[\max\{0,h-2\} \leq d_j(\x)\leq h]\leq \Pr[d_j(\x)\geq U_j]<n^{-5.33}$ as well. Since, trivially, $h-d_t(\x)\leq n$, the contribution of the at most $n^4$ terms of the sum in which either $h$ or $g$ does not belong to the comfort zone $Z_t$ or $h-2$ lies above the comfort zone $Z_j$ is at most $n^4\cdot n\cdot n^{-5.33} < 1$. Hence, equation (\ref{eq:THE-sum}) becomes
\begin{align}\nonumber
\E[(h-d_w(\x))\one\{A\}]
&\leq 1+\sum_{h=L_t}^{U_t}{\sum_{g=L_t}^{h}{(h-g) \cdot \Pr[d_t(\x)= g]\sum_{i\in N_t:h\in Z_i}{\Pr[d_i(\x)=h]}}}\\\label{eq:THE-sum-med}
& \quad\quad {{\cdot \sum_{j\in N_{i,t}:h-2\in Z_j}{\Pr[\max\{0,h-2\} \leq d_j(\x)\leq h]}\prod_{k\in N_{i,j,t}}{\Pr[d_k(\x)\leq h]}}}
\end{align}
Our aim in the following is to evalute each term in the sum at the RHS of (\ref{eq:THE-sum-med}). To do so, we will need three auxiliary technical lemmas. The proofs of the first two follow easily by applying Chernoff bounds.

\begin{lemma}\label{lem:zone-width}
	For the boundaries of the comfort zone $Z_t$ we have $U_t\leq \mu_t+4\sqrt{\xi_t\ln{n}}$ and $L_t\geq \mu_t-4\sqrt{\xi_t\ln{n}}$.  
\end{lemma}

\begin{proof}
	If $\mu_t\geq n/2$, by applying the tail inequality (\ref{eq:okamoto}) from Lemma~\ref{lem:chernoff}, we get
	\begin{align*}
	\Pr[d_t(\x)\geq \mu_t+4\sqrt{\xi_t \ln{n}}] &\leq \exp\left(-\frac{(4\sqrt{\xi_t\ln{n}})^2n}{2\xi_t (n-\xi_t)}\right) \leq n^{-8}.
	\end{align*}
	If $\mu_t<n/2$, observe that $4\sqrt{\xi_t\ln{n}}\leq \mu_t$, by our assumption $\xi_t\geq 8200\ln{n}$. Hence, by applying the tail inequality (\ref{eq:chernoff-upper}) from Lemma~\ref{lem:chernoff}, we get
	\begin{align*}
	\Pr[d_t(\x)\geq \mu_t+4\sqrt{\xi_t \ln{n}}] &\leq \exp\left(-\frac{(4\sqrt{\xi_t\ln{n}})^2}{3\mu_t}\right) \leq n^{-5.33}.
	\end{align*}
	The bounds on $U_t$ follows by its definition.
	
	Similarly, if $\mu_t\leq n/2$, by applying the tail inequality (\ref{eq:okamoto-lower}) from Lemma~\ref{lem:chernoff}, we get
	\begin{align*}
	\Pr[d_t(\x)\leq \mu_t-4\sqrt{\xi_t \ln{n}}] &\leq \exp\left(-\frac{(4\sqrt{\xi_t\ln{n}})^2n}{2\xi_t(n-\xi_t)}\right) \leq n^{-8}.
	\end{align*}
	If $\mu_t> n/2$, observe that $\mu_t-4\sqrt{\xi_t\ln{n}} \geq 2\mu_t-n$, by our assumption $\xi_t=n-\mu_t\geq 8200\ln{n}$. Hence, by applying the tail inequality (\ref{eq:chernoff-lower}) from Lemma~\ref{lem:chernoff}, we get
	\begin{align*}
	\Pr[d_t(\x)\leq \mu_t-4\sqrt{\xi_t \ln{n}}] &\leq \exp\left(-\frac{(4\sqrt{\xi_t\ln{n}})^2}{3\mu_t}\right) \leq n^{-5.33}.
	\end{align*}
	Again, the bound on $L_k$ follows by its definition.
\end{proof}

We will say that the comfort zones $Z_k$ and $Z_{k'}$ {\em almost intersect} if $L_{k'}-U_k\leq 2$ or $L_k-U_{k'}\leq 2$. For example, since $h\in Z_t$ and $h-2\in Z_j$, the two comfort zones $Z_t$ and $Z_j$ almost intersect.

\begin{lemma}\label{lem:mu-xi-ratio}
	If two comfort zones $Z_k$ and $Z_{k'}$ almost intersect, then $\frac{3}{4} \mu_k\leq \mu_{k'}\leq \frac{4}{3}\mu_k$ and $\frac{16}{25}\xi_k\leq \xi_{k'} \leq \frac{25}{16}\xi_k$. 
\end{lemma}

\begin{proof}
	Without loss of generality, assume that $\mu_k\leq \mu_{k'}$; the other case is completely symmetric. Then, $L_{k'}-U_{k}\leq 2$, which, using the facts $\xi_k,\xi_{k'}\leq \mu_{k'}$ and $\mu_{k'}\geq 8200\ln{n}$ as well as Lemma~\ref{lem:zone-width}, implies that
	\begin{align*}
	\mu_{k'} &\leq 2+\mu_{k}+4\sqrt{\xi_k\ln{n}}+4\sqrt{\xi_{k'}\ln{n}} \leq 2+\mu_{k}+8\sqrt{\mu_{k'}\ln{n}} \leq \mu_{k}\left(\frac{1}{4100}+1\right)+\frac{8\mu_{k'}}{\sqrt{8200}},
	\end{align*} 
	which clearly implies that $\mu_{k'}\leq \frac{4}{3}\mu_k$. 
	
	Also, observe that 
	\begin{align*}
	\max\{\xi_{k'},\xi_k\}-\min\{\xi_{k'},\xi_{k}\} &\leq \mu_{k'}-\mu_k \leq 2+4\sqrt{\xi_k\ln{n}}+4\sqrt{\xi_{k'}\ln{n}}\\
	&\leq 2+8\sqrt{\max\{\xi_k,\xi_{k'}\}\ln{n}}\\
	&\leq \left(\frac{1}{4100}+\frac{8}{\sqrt{8200}}\right) \max\{\xi_{k},\xi_{k'}\},
	\end{align*}
	which implies that $\max\{\xi_{k'},\xi_{k}\} \leq \frac{25}{16} \min\{\xi_{k'},\xi_{k}\}$ as desired.
\end{proof}

We now prove the most important technical lemma in our analysis. 

\begin{lemma}\label{lem:technical}
	Let $\ell\in\{0,1,2\}$ and $h$ be such that $h\in Z_t$ and $h-2\in Z_k$ for some agent $k$. Then,
	$$\Pr[d_k(\x)=h-\ell]\leq 264 e\sqrt{\frac{\ln{n}}{\xi_t}}\cdot \Pr[d_k(\x)>h].$$ 
\end{lemma}

\begin{proof}
	By the definition of the binomial distribution, we have 
	\begin{align*}
	\Pr[d_k(\x)=z] &={n\choose z} p_k^{z}(1-p_k)^{n-z}
	\end{align*} 
	for every integer $z$ with $0\leq z\leq n$. Let $x$ be any positive integer with $x\leq \mu_k+4\sqrt{\xi_k\ln{n}}$. For every integer $y>x$, we have
	\begin{align}\nonumber
	\frac{\Pr[d_k(\x)=x]}{\Pr[d_k(\x)=y]} &=\frac{{n \choose x} p_k^{x}(1-p_k)^{n-x}}{{n \choose y} p_k^{y}(1-p_k)^{n-y}}\\\nonumber
	&= \frac{(x+1)\cdot (x+2)\cdot  ... \cdot y}{(n-y+1)\cdot (n-y+2)\cdot... \cdot (n-x)}\cdot \frac{(1-p_k)^{y-x}}{p_k^{y-x}}\\\nonumber
	&= \frac{\left(1+\frac{x-\mu_k+1}{\mu_k}\right)\cdot \left(1+\frac{x-\mu_k+2}{\mu_k}\right)\cdot... \cdot \left(1+\frac{y-\mu_k}{\mu_k}\right)}{\left(1-\frac{y-\mu_k-1}{n-\mu_k}\right)\cdot \left(1-\frac{y-\mu_k}{n-\mu_k}\right)\cdot ... \cdot \left(1-\frac{x-\mu_k}{n-\mu_k}\right)}\\\nonumber
	&\leq \frac{\left(1+\frac{y-\mu_k}{\mu_k}\right)^{y-x}}{\left(1-\frac{y-\mu_k-1}{n-\mu_k}\right)^{y-x}}\\\label{eq:exp-bound}
	& \leq \exp\left(\frac{(y-\mu_k)(y-x)}{\mu_k}+\frac{(y-\mu_k+1)(y-x)}{n-y+1}\right)
	\end{align}
	The first inequality follows since $x<y$. In the second inequality, we have used the properties $1+z\leq e^z$ for $z\in \R$ and, consequently, $\frac{1}{1-z}=1+\frac{z}{1-z}\leq \exp(\frac{z}{1-z})$ for $z\not=1$.
	
	We now use inequality (\ref{eq:exp-bound}) to argue that by selecting $y$ such that $y>h-\ell$ and 
	\begin{align}\label{eq:y-vs-x}
	(y-\mu_k+1)(y-h+\ell) &\leq \frac{3}{11}(\xi_t-y+\mu_k), 
	\end{align}
	we get 
	\begin{align}\label{eq:bound-e}
	\frac{\Pr[d_k(\x)=h-\ell]}{\Pr[d_k(\x)=y]}&\leq e.
	\end{align}
	Recall that $Z_k$ and $Z_t$ almost intersect. Hence, we have $\mu_k\geq 3\mu_t/4$ (by Lemma~\ref{lem:mu-xi-ratio}), and (\ref{eq:y-vs-x}) yields
	\begin{align}\label{eq:8-11}
	\frac{(y-\mu_k)(y-h+\ell)}{\mu_k} &\leq \frac{3}{11}\cdot \frac{\xi_t-y+\mu_k}{\mu_k} \leq \frac{4}{11}\frac{2\mu_t-y}{\mu_t}\leq \frac{8}{11}.
	\end{align}
	Furthermore, using again (\ref{eq:y-vs-x}), and the inequalities $\mu_k\leq \mu_t$ and $\xi_t\leq n-\mu_t$, we get
	\begin{align}\label{eq:3-11}
	\frac{(y-\mu_k+1)(y-d+\ell)}{n-y+1} &\leq \frac{3}{11}\cdot \frac{\xi_t-y+\mu_k}{n-y+1}\leq  \frac{3}{11}\cdot  \frac{\xi_t-y+\mu_t}{n-y+1}\leq \frac{3}{11}\cdot  \frac{n-y}{n-y+1} \leq \frac{3}{11}.
	\end{align}
	Inequality (\ref{eq:bound-e}) now follows by inequalities (\ref{eq:exp-bound}), (\ref{eq:8-11}), and (\ref{eq:3-11}).
	
	Solving inequality (\ref{eq:y-vs-x}), we get that the range of values for $y$ so that (\ref{eq:bound-e}) is true satisfies 
	\begin{align*}
	h-\ell< y \leq \frac{h-\ell+\mu_k-\frac{14}{11}+\sqrt{(h-\ell-\mu_k)^2+\frac{16}{11}(h-\ell-\mu)+\frac{12}{11}\xi_t+\frac{196}{121}}}{2}.
	\end{align*}
	Hence, the number of integer values for $y$ so that $y>h-\ell$ and (\ref{eq:y-vs-x}) holds is at least 
	\begin{align}\nonumber
	& \frac{h-\ell+\mu_k-\frac{14}{11}+\sqrt{(h-\ell-\mu_k)^2+(d-\ell-\mu_k)+\frac{12}{11}\xi_t+\frac{196}{121}}}{2}-h+\ell-2\\\label{eq:x-function}
	&=\frac{\sqrt{(h-\ell-\mu_k)^2+(d-\ell-\mu_k)+\frac{12}{11}\xi_t+\frac{196}{121}}-(h-\ell-\mu_k+96/11)}{2}+2.
	\end{align}
	The derivative of the quantity at the RHS of (\ref{eq:x-function}) with respect to $h$ is
	\begin{align*}
	\frac{2(h-\ell-\mu_k)+1}{4\sqrt{(h-\ell-\mu_k)^2+(h-\ell-\mu_k)+\frac{12}{11}\xi_t+\frac{196}{121}}}-\frac{1}{2}<0,
	\end{align*}
	i.e., it is decreasing. Since $h-\ell\in Z_k$ and $Z_k$ and $Z_t$ almost intersect, using Lemmas~\ref{lem:zone-width} and~\ref{lem:mu-xi-ratio} we have $h-\ell-\mu_k\leq 4\sqrt{\xi_k\ln{n}}\leq 5\sqrt{\xi_t\ln{n}}$. Hence, we can bound the RHS of (\ref{eq:x-function}) as follows:
	\begin{align*}
	& \frac{\sqrt{(h-\ell-\mu_k)^2+(d-\ell-\mu_k)+\frac{12}{11}\xi_t+\frac{196}{121}}-\left(h-\ell-\mu_k+\frac{96}{11}\right)}{2}+2\\
	&=\frac{(h-\ell-\mu_k)^2+(d-\ell-\mu_k)+\frac{12}{11}\xi_t+\frac{196}{121}-\left(h-\ell-\mu_k+\frac{96}{11}\right)^2}{2(\sqrt{(h-\ell-\mu_k)^2+(d-\ell-\mu_k)+\frac{12}{11}\xi_t+\frac{196}{121}}+(h-\ell-\mu_k+96/11))}+2\\
	&=\frac{12\xi_t-181(h-\ell-\mu_k)-820}{22\left(\sqrt{(h-\ell-\mu_k)^2+(d-\ell-\mu_k)+\frac{12}{11}\xi_t+\frac{196}{121}}+h-\ell-\mu_k+96/11\right)}+2\\
	&\geq \frac{12\xi_t-905\sqrt{\xi_t\ln{n}}-820}{22\left(\sqrt{25\xi_t\ln{n}+5\sqrt{\xi_t\ln{n}}+\frac{12}{11}\xi_t+\frac{196}{121}}+5\sqrt{\xi_t\ln{n}}+\frac{96}{11}\right)} \geq \frac{1}{264}\sqrt{\frac{\xi_t}{\ln{n}}}+2.
	\end{align*}
	In the second inequality, we have used $905\sqrt{\xi_t\ln{n}}\leq 10\xi_t$ and $820\leq \xi_t$ to bound the numerator by $\xi_t$ (recall that $\xi_t\geq 8200\ln{n}$), while the parenthesis in the denominator is clearly at most $12\sqrt{\xi_t\ln{n}}$.
	
	Now, let $\ell\in \{0,1,2\}$ and $r=\left\lceil\frac{1}{264}\sqrt{\frac{\xi_t}{\ln{n}}}\right\rceil+2$. By the discussion above, for $x=h-\ell$ we have 
	\begin{align*}
	\Pr[d_k(\x)=h-\ell] &\leq e \Pr[d_k(\x)=y]
	\end{align*}
	for $y=h-\ell+1, h-\ell+2, ..., h-\ell+r+2$. By summing these inequalities for $y=h+1, ...,h+r$, we get
	\begin{align*}
	r\Pr[d_k(\x)=h-\ell] &\leq e \sum_{y=h+1}^{r}{\Pr[d_k(\x)=y]} \leq e\Pr[d_k(\x)>h]
	\end{align*}
	and, equivalently, 
	\begin{align*}
	\Pr[d_k(\x)=h-\ell] &\leq \frac{e}{r}\Pr[d_k(\x)>h]\leq 264e\sqrt{\frac{\ln{n}}{\xi_t}}\Pr[d_k(\x)>h].
	\end{align*}
	The lemma follows.
\end{proof}

We are ready to complete the proof of Theorem~\ref{thm:plurality-with-default}. For $h\in Z_t$, using Lemma~\ref{lem:zone-width}, we have
\begin{align}\label{eq:d-star-minus-dt}
\sum_{h\in Z_t}{\sum_{g\in Z_t}{(h-g)\Pr[d_t(\x)=g]}} &\leq  \sum_{h\in Z_t}{8\sqrt{\xi_t\ln{n}} \sum_{g\in Z_t}{\Pr[d_t(\x)=g]}} \leq 64\xi_t\ln{n}.
\end{align}
Furthermore, Lemma~\ref{lem:technical} yields
\begin{align}\nonumber
& \sum_{i\in N_t}{\sum_{j\in N_{i,t}}{\Pr[d_i(\x)=h]\Pr[\max\{0,h-2\}\leq d_j(\x)\leq h]\prod_{k\in N_{i,j,t}}{\Pr[d_k(\x)\leq h]}}}\\\nonumber
& \leq 3\left(264\cdot e \sqrt{\frac{\ln{n}}{\xi_t}}\right)^2 \sum_{i\in N_t}{\sum_{j\in N_{i,t}}{\Pr[d_i(\x)>h]\Pr[d_j(\x)> h]\prod_{k\in N_{i,j,t}}{\Pr[d_k(\x)\leq h]}}}\\\label{eq:prob-two-higher}
&\leq 209088\cdot e^2\cdot \frac{\ln{n}}{\xi_t}.
\end{align}
since the last double sum is simply the probability that exactly two agents have degree higher than $h$ (and, hence, has value at most $1$).
Using (\ref{eq:d-star-minus-dt}) and (\ref{eq:prob-two-higher}), equation (\ref{eq:THE-sum-med}) yields $\E[(h-d_w(\x))\one\{A\}]\leq 1+10^8\cdot \ln^2n$ and the proof of Theorem~\ref{thm:plurality-with-default} is now complete.\qed

\section{A lower bound for AVD}\label{sec:lb-AVD}

In this section, we prove the following lower bound for the uniform domain.

\begin{theorem}\label{thm:lb-uniform}
	When applied on uniform instances with $p=1/2$, the AVD mechanism has expected additive approximation $\Omega(\ln n)$.
\end{theorem}

With uniform instances, the in-degree of each node follows the binomial distribution. In the proof of Theorem~\ref{thm:lb-uniform}, we use the random variables $B$ and $B'$ following the distributions $\bin(n,1/2)$ and $\bin(n-1,1/2)$, respectively. We also assume that $n$ is sufficiently large.

Let $U$ be the lowest integer $c$ such that $\Pr[B>c]\leq \frac{1}{3e^2n\sqrt{6}}$. Similarly, let $L$ be the lowest integer $c$ such that $\Pr[B>c]< \frac{1}{n\sqrt{2}}$. Consider the following event $D$:
\begin{itemize}
	\item The default node has degree at most $n/2$,
	\item two non-default nodes (called the {\em potential winners}) have the same in-degree $d \in [L+1,U]$, without counting the edges between them, 
	\item the remaining non-default nodes (called the {\em losers}) have in-degree at most $d-1$.
\end{itemize}
Then, AVD returns the default node as a winner and the additive approximation is at least $L-n/2$. We will show that $\E[(L-n/2)\one\{D\}]$ is $\Omega(\ln{n})$, proving the lemma. In particular, we will use the inequality
\begin{align}\label{eq:D}
\E[(L-n/2)\one\{D\}] & \geq \frac{1}{2}(L-n/2) \sum_{d=L+1}^U{{n\choose 2}\Pr[B'=d]^2 \cdot \Pr[B\leq d-1]^{n-2}}.
\end{align}
The RHS of equation (\ref{eq:D}) is the product of the lower bound of the additive approximation $L-n/2$ when event $D$ happens, with $1/2$ which is (a lower bound on) the probability that the default node has degree at most $n/2$, and with the probability $\Pr[B'=d]^2$ that the two potential winners have degree $d$ (ignoring the edges between them) and the probability $\Pr[B\leq d-1]^{n-2}$ that the losers have degree at most $d-1$, for all the ${n\choose 2}$ selections of the two potential winners. 

We will make use of a series of lemmas to bound the several quantities that appear in the RHS of equation (\ref{eq:D}).

\begin{lemma}\label{lem:lb-step1}
	$L \geq \frac{n}{2}+\sqrt{\frac{n\ln{n}}{6}}$.	
\end{lemma}

\begin{proof}
By applying Corollary~\ref{cor:inverse} to the random variable $B\sim\bin(n,1/2)$ with $\delta=\sqrt{\frac{\ln{n}}{6n}}$ (observe that $\delta\leq 1/10$ since $n$ is large), we have 
$\Pr\left[B\geq \frac{n}{2}+\sqrt{\frac{n\ln{n}}{6}}\right] \geq \frac{1}{n\sqrt{2}}$
for the random variable $B\sim \bin(n,1/2)$. The lemma follows by the definition of $L$.
\end{proof}

\begin{lemma}\label{lem:chernoff-U}
	$U \leq \frac{n}{2}+\sqrt{n\ln{n}}$.
\end{lemma}

\begin{proof}
	A simple application of the Chernoff bound (inequality (\ref{eq:okamoto}) from Lemma~\ref{lem:chernoff}) to the binomial random variable $B\sim \bin(n,1/2)$ yields $\Pr\left[B\geq \frac{n}{2}+\sqrt{n\ln{n}}\right]\leq \frac{1}{n^2}\leq \frac{1}{3e^2\sqrt{6}}$. The lemma then follows by the definition of $U$.
\end{proof}

\begin{lemma}\label{lem:lb-step2}
	For the random variable $B\sim\bin(n,1/2)$, it holds that
	\begin{align*}
	\Pr[B=x] &\geq \frac{2L-n}{n}\cdot \Pr[B\geq x],
	\end{align*}
	for every integer $x \geq L$.
\end{lemma}

\begin{proof}
	Consider integers $x,y$ with $L\leq x\leq y$. By the definition of the binomial distribution $\bin(n,1/2)$, we have
	\begin{align*}
	\frac{\Pr[B=y]}{\Pr[B=x]} &=\frac{{n\choose y}}{{n\choose x}} =\frac{x! (n-x)!}{y! (n-y)!}\leq \left(\frac{n-x}{x}\right)^{y-x}\leq \left(\frac{n-L}{L}\right)^{y-x}.
	\end{align*}
	Hence,
	\begin{align*}
	\Pr[B\geq x] &= \sum_{y=x}^{n}\Pr[B=y] \leq \Pr[B=x] \cdot \sum_{y=x}^n{\left(\frac{n-L}{L}\right)^{y-x}}\\
	&\leq \frac{L}{2L-n} \cdot \Pr[B=x] \leq \frac{n}{2L-n} \cdot \Pr[B=x],
	\end{align*}
	and the lemma follows by rearranging.
\end{proof}

The proof of the next lemma follows a similar roadmap with the proof of Lemma~\ref{lem:technical} in Section~\ref{sec:apriori} but is considerably simpler. In the proof, we will use the following claim, which will also be useful later. The proof follows easily by the definition of the binomial distribution.

\begin{claim}\label{claim:expo}
	For the random variable $B\sim\bin(n,1/2)$ and integers $x$ and $y$ with $x\leq y$, it holds that 
	\begin{align*}
	\Pr[B=x] &\leq \left(\frac{y}{n-y}\right)^{y-x}\cdot \Pr[B=y].
	\end{align*}
\end{claim}

\begin{proof}
By the definition of the binomial distribution $\bin(n,1/2)$, we have 
\begin{align}\nonumber
\frac{\Pr[B=x]}{\Pr[B=y]} &= \frac{{n \choose x}}{{n \choose y}}=\frac{y!(n-x)!}{x!(n-y)!} \leq \left(\frac{y}{n-y}\right)^{y-x}. \qedhere
\end{align} 	
\end{proof}

\begin{lemma}\label{lem:pdf-vs-one-minus-cdf-p-equals-half}
	For the random variable $B\sim \bin(n,1/2)$, it holds that
	\begin{align*}
	\Pr[B=U] &\leq \frac{8}{3e\sqrt{6}} \frac{\sqrt{\ln{n}}}{n^{3/2}}.
	\end{align*}
\end{lemma}

\begin{proof}
	By Claim~\ref{claim:expo}, we have 
	\begin{align}\label{eq:exponential-bound}
	\frac{\Pr[B=U]}{\Pr[B=y]} &\leq \left(\frac{y}{n-y}\right)^{y-U}\leq \exp\left(\frac{(y-U)(2y-n)}{n-y}\right)
	\end{align} 
	for every integer $y>U$. The second inequality follows since $e^z\geq 1+z$ for $z\in \R$. By selecting $y$ such that 
	\begin{align}\label{eq:range-ineq}
	\frac{(y-U)(2y-n)}{n-y} &\leq 1,
	\end{align}
	we get 
	\begin{align}\label{eq:e-bound}
	\frac{\Pr[B=U]}{\Pr[B=y]} &\leq e.
	\end{align}
	Solving inequality (\ref{eq:range-ineq}), we get that the range of values for $y$ so that (\ref{eq:e-bound}) is true satisfies
	\begin{align*}
	U&<y\leq \frac{2U+n-1+\sqrt{(2U-n)^2+6n-4U+1}}{4}.
	\end{align*}
	Hence, the number of integer values for $y$ so that $y>U$ and (\ref{eq:range-ineq}) is satisfied is at least
	\begin{align}\nonumber
	&\frac{2U+n-1+\sqrt{(2U-n)^2+6n-4U+1}}{4}-U-1\\\label{eq:range-int}
	&= \frac{\sqrt{(2U-n)^2+6n-4U+1}-2U+n-5}{4}.
	\end{align}
	Now observe that the quantity at the RHS of (\ref{eq:range-int}) is non-increasing with respect to $U$ since its derivative  
	\begin{align*}
	\frac{2U-n-1}{2\sqrt{(2U-n)^2+6n-4U+1}}-\frac{1}{2}
	\end{align*}
	is non-positive. So, we can bound the RHS of (\ref{eq:range-int}) from below using the upper bound on $U$ from Lemma~\ref{lem:chernoff-U}. We get
	\begin{align*}
	&\frac{\sqrt{(2U-n)^2+6n-4U+1}-2U+n-5}{4}\\
	&= \frac{16n-24U-25}{4\left(\sqrt{(2U-n)^2+6n-4U+1}+2U-n+5\right)}\\
	&\geq \frac{4n-24\sqrt{n\ln{n}}-25}{4\left(\sqrt{4n\ln{n}+4n-4\sqrt{n\ln{n}}+1}+2\sqrt{n\ln{n}}-n+5\right)}\\
	&\geq \frac{1}{8}\sqrt{\frac{n}{\ln{n}}}.
	\end{align*}
	In the last inequality, we have used $24\sqrt{n\ln{n}}+25\leq n$ to lower-bound the numerator by $3n$ and $5\leq \sqrt{n\ln{n}}$ to upper-bound the parenthesis in the denominator by $6\sqrt{n\ln{n}}$.
	
	Now, let $r=\left\lceil\frac{1}{8}\sqrt{\frac{n}{\ln{n}}}\right\rceil$. Multiplying inequality (\ref{eq:e-bound}) by $1/r$ and summing these inequalities for $y=U+1, ..., U+r$, we have
	\begin{align*}
	\Pr[B=U] &\leq \frac{e}{r} \sum_{y=U+1}^{U+r}{\Pr[B=y]} \leq \frac{e}{r} \cdot \Pr[B>U]\leq \frac{8}{3e\sqrt{6}} \frac{\sqrt{\ln{n}}}{n^{3/2}},
	\end{align*} 
	as desired. The last inequality follows by the definition of $U$.
\end{proof}

\begin{lemma}\label{lem:lb-step3}
	$U-L \geq \frac{1}{6}\sqrt{\frac{n}{\ln{n}}}$.
\end{lemma}

\begin{proof}
	Let $B\sim\bin(n,1/2)$. Using the definition of $U$ and $L$, we have
\begin{align}\nonumber
	\frac{2}{3n} &\leq \frac{1}{n\sqrt{2}}-\frac{1}{3e^2n\sqrt{6}} \leq \Pr[L\leq B\leq U]\\\nonumber
	&\leq \sum_{x=L}^U{\left(\frac{U}{n-U}\right)^{U-x}}\cdot \Pr[B=U]\\\nonumber
	&\leq \left(\frac{U}{n-U}\right)^{U-L+1}\cdot \frac{n-U}{2U-n}\cdot \frac{8}{3e\sqrt{6}}\frac{\sqrt{\ln{n}}}{n^{3/2}} \nonumber \\\nonumber
	&\leq \exp\left(\frac{2U-n}{n-U}(U-L+1)\right)\cdot \frac{n-L}{2L-n}\cdot \frac{8}{3e\sqrt{6}}\frac{\sqrt{\ln{n}}}{n^{3/2}}\\\label{eq:U-L}
	&\leq \exp\left(\frac{2U-n}{n-U}(U-L+1)\right)\cdot \frac{2}{3en}.
	\end{align}
	The first inequality is obvious, while the second one uses the definition of $U$ and $L$ (recall that $\Pr[B\geq L]\geq \frac{1}{n\sqrt{2}}$ and $\Pr[B>U]\leq \frac{1}{4e^2n\sqrt{3}}$). The third inequality follows by Claim~\ref{claim:expo}. The fourth inequality follows by Lemma~\ref{lem:pdf-vs-one-minus-cdf-p-equals-half}, the fifth one follows since $U\geq L$ and by the definition of $U$, and the sixth one follows by Lemma~\ref{lem:lb-step1} and the fact $L\geq n/2$.
	
	By Lemma~\ref{lem:chernoff-U} and due to the high value of $n$, $2n-3U\geq n/3$. Hence, inequality (\ref{eq:U-L}) implies that 
	\begin{align*}
	U-L &\geq \frac{2n-3U}{2U-n} \geq \frac{1}{6}\sqrt{\frac{n}{\ln{n}}},
	\end{align*}
	as desired.
\end{proof}

\begin{lemma}\label{lem:n-1}
	Let $B\sim\bin(n,1/2)$ and $B'\sim\bin(n-1,1/2)$. For every integer $x\in [L+1,U]$, $\Pr[B'=x]\geq \frac{2}{3} \Pr[B=x]$.
\end{lemma}
\begin{proof}
	By the definition of the binomial distribution, Lemma~\ref{lem:chernoff-U}, and the facts that $x\leq U$ and that $n$ is large (the last two imply that $x\leq 2n/3$), we have 
	\begin{align*}
	\Pr[B'=x] &={n-1\choose x}2^{-n+1}= 2\frac{n-x}{n}{n\choose x}2^{-n}\geq \frac{2}{3}\Pr[B=x]. \qedhere
	\end{align*}
\end{proof}

We are now ready to bound $\E[(L-n/2)\one\{D\}]$ from below. Using equation (\ref{eq:D}), and the lemmas above, we have
\begin{align*}
\E(L-n/2)\one\{D\} &\geq \frac{1}{2} \left(L-\frac{n}{2}\right) \cdot \sum_{d=L+1}^U{{n\choose 2} \Pr[B'=d]^2\Pr[B\leq d-1]^{n-2}}\\
&\geq \frac{2}{9} \left(L-\frac{n}{2}\right) \cdot \sum_{d=L+1}^U{{n\choose 2} \Pr[B=d]^2\Pr[B\leq d-1]^{n-2}}\\
&\geq \frac{8}{9n^2} \left(L-\frac{n}{2}\right)^3 \sum_{d=L+1}^U{{n\choose 2} \Pr[B\geq d]^2\Pr[B\leq d-1]^{n-2}}\\
&\geq \frac{8}{9n^2} \left(L-\frac{n}{2}\right)^3(U-L) {n\choose 2} \left(\frac{1}{3e^2n\sqrt{6}} \right)^2 \left(1-\frac{1}{n\sqrt{2}}\right)^{n-2}\\
&\geq \frac{1}{6561 \, e^{5}\sqrt{6}} \cdot\ln{n}.
\end{align*}
The second inequality follows by Lemma~\ref{lem:n-1}. The third inequality follows by Lemma~\ref{lem:lb-step2}. The fourth inequality follows by the definition of $L$ and $U$. Finally, the fifth inequality follows by Lemma~\ref{lem:lb-step1}, Lemma~\ref{lem:lb-step3}, and the fact $\left(1-\frac{1}{n\sqrt{2}}\right)^{n-2}\geq 1/e$.  Theorem~\ref{thm:lb-uniform} follows. \qed

\section{Open problems}\label{sec:open}
Our polylogarithmic upper bound in Section~\ref{sec:apriori} shows
that prior information can yield dramatic improvements on the
performance of simple impartial selection mechanisms. It also gives
hope that AVD could be similarly efficient for the more general
opinion poll instances. Unfortunately, this is not true as the
following counter-example indicates.

Indeed, starting from a uniform instance with $n+1$ nodes of
popularity $p=1/2$, we add a new copy $j'$ for each node $j$. Also,
for every edge $(i,j)$ realized in the original instance, we add the
edge $(i,j')$. In this way, we construct opinion poll instances with
$2(n+1)$ nodes, where no node can ever beat all the other
nodes. Hence, in such instances, AVD will behave as the constant
mechanism in the original instance and will always return the default
node as winner. By applying Theorem~\ref{LB:GeneralmodelLB} we
obtain the following negative result for AVD. 

\begin{theorem}\label{thm:lb-opinion-poll}
  When applied on opinion poll instances, the AVD mechanism has
  expected additive approximation $\Omega(\sqrt{n\ln{n}})$.
\end{theorem}

We should note that the above construction is fragile, in the sense
that it exploits a very specific aspect of the mechanism. So, still,
the quest of designing deterministic mechanisms that achieve
polylogarithmic additive approximation in the opinion poll model is
very important and challenging. A starting step could be to restrict
our attention to the instances considered in~\cite{moulin2013}, in
which every voter approves exactly one other candidate.

Finally, throughout the paper, we have assumed that the prior
information is reliable. This should not be expected to be the case in
practice. We expect that our results on the constant mechanism still
hold if we have a rough estimate of the highest in-degree. Highest
accuracy seems to be necessary to recover our polylogarithmic upper
bound though. This issue is also related to the strengths of
prior-independent mechanisms (e.g., see Section 4.3
of~\cite{hartline2013bayesian}) and needs to be investigated further.


\section*{Acknowledgements}
This work was partially supported by COST Action 16228 ``European Network for Game Theory''.

\bibliographystyle{plain}
\bibliography{thebib}

\begin{thebibliography}{10}

\bibitem{alon11}
Noga Alon, Felix Fischer, Ariel Procaccia, and Moshe Tennenholtz.
\newblock Sum of us: Strategyproof selection from the selectors.
\newblock In {\em Proceedings of the 13th Conference on Theoretical Aspects of
  Rationality and Knowledge (TARK)}, pages 101--110, 2011.

\bibitem{Ash90}
Robert~B. Ash.
\newblock {\em Information Theory}.
\newblock Courier Corporation, 1990.

\bibitem{aziz2019strategyproof}
Haris Aziz, Omer Lev, Nicholas Mattei, Jeffrey~S. Rosenschein, and Toby Walsh.
\newblock Strategyproof peer selection using randomization, partitioning, and
  apportionment.
\newblock {\em Artificial Intelligence}, 275:295--309, 2019.

\bibitem{babichenko2018incentive}
Yakov Babichenko, Oren Dean, and Moshe Tennenholtz.
\newblock Incentive-compatible diffusion.
\newblock In {\em Proceedings of the 27th International Conference on World
  Wide Web (WWW)}, pages 1379--1388, 2018.

\bibitem{babichenko2020incentive}
Yakov Babichenko, Oren Dean, and Moshe Tennenholtz.
\newblock Incentive-compatible selection mechanisms for forests.
\newblock In {\em Proceedings of the 21st ACM Conference on Economics and
  Computation (EC)}, page 111–131, 2020.

\bibitem{bjelde2017}
Antje Bjelde, Felix Fischer, and Max Klimm.
\newblock Impartial selection and the power of up to two choices.
\newblock {\em ACM Transactions on Economics and Computation}, 5(4):21, 2017.

\bibitem{bollobas2001random}
Béla Bollobás.
\newblock {\em Random Graphs}.
\newblock Cambridge Studies in Advanced Mathematics. Cambridge University
  Press, 2nd edition, 2001.

\bibitem{bousquet2014}
Nicolas Bousquet, Sergey Norin, and Adrian Vetta.
\newblock A near-optimal mechanism for impartial selection.
\newblock In {\em Proceedings of the 10th International Conference on Web and
  Internet Economics (WINE)}, pages 133--146, 2014.

\bibitem{BSW94}
J.J.A.M. Brands, F.W. Steutel, and R.J.G. Wilms.
\newblock On the number of maxima in a discrete sample.
\newblock {\em Statistics and Probability Letters}, 20(3):209--217, 1994.

\bibitem{caragiannis2019impartial}
Ioannis Caragiannis, George Christodoulou, and Nicos Protopapas.
\newblock Impartial selection with additive approximation guarantees.
\newblock In {\em Proceedings of the 12th International Symposium on
  Algorithmic Game Theory (SAGT)}, pages 269--283, 2019.

\bibitem{declipper2008}
Geoffroy de~Clippel, Herv{\'e} Moulin, and Nicolaus Tideman.
\newblock {Impartial division of a dollar}.
\newblock {\em Journal of Economic Theory}, 139(1):176--191, 2008.

\bibitem{ES96}
Bennett Eisenberg and Gilbert Stengle.
\newblock Minimizing the probability of a tie for first place.
\newblock {\em Journal of Mathematical Analysis and Applications},
  198(2):458--472, 1996.

\bibitem{ESS93}
Bennett Eisenberg, Gilbert Stengle, and Gilbert Strang.
\newblock The asymptotic probability of a tie for first place.
\newblock {\em The Annals of Applied Probability}, 3(3):731--745, 1993.

\bibitem{erdos1977chromatic}
Paul Erd{\H{o}}s and Robin~J. Wilson.
\newblock On the chromatic index of almost all graphs.
\newblock {\em Journal of Combinatorial Theory, Series B}, 23(2-3):255--257,
  1977.

\bibitem{fischer2015}
Felix Fischer and Max Klimm.
\newblock Optimal impartial selection.
\newblock {\em SIAM Journal on Computing}, 44(5):1263--1285, 2015.

\bibitem{frieze2016introduction}
Alan Frieze and Micha{\l} Karo{\'n}ski.
\newblock {\em Introduction to random graphs}.
\newblock Cambridge University Press, 2016.

\bibitem{hartline2013bayesian}
Jason~D. Hartline.
\newblock {\em Bayesian Mechanism Design}.
\newblock Foundations and trends in theoretical computer science. Now
  Publishers, 2013.

\bibitem{H63}
Wassily Hoeffding.
\newblock Probability inequalities for sums of bounded random variables.
\newblock {\em Journal of the American Statistical Association},
  58(301):13--30, 1963.

\bibitem{moulin2013}
Ron Holzman and Hervé Moulin.
\newblock Impartial nominations for a prize.
\newblock {\em Econometrica}, 81(1):173--196, 2013.

\bibitem{kahng2018ranking}
Anson Kahng, Yasmine Kotturi, Chinmay Kulkarni, David Kurokawa, and Ariel~D
  Procaccia.
\newblock Ranking wily people who rank each other.
\newblock In {\em Proceedings of the 32nd AAAI Conference on Artificial
  Intelligence (AAAI)}, pages 1087--1094, 2018.

\bibitem{Kurokawa2015}
David Kurokawa, Omer Lev, Jamie Morgenstern, and Ariel~D. Procaccia.
\newblock {Impartial peer review}.
\newblock In {\em Proceedings of the 24th International Joint Conference on
  Artificial Intelligence (IJCAI)}, pages 582--588, 2015.

\bibitem{LS10}
Jean-Francois Laslier and M.~Remzi Sanver, editors.
\newblock {\em Handbook on Approval Voting}.
\newblock Springer, 2010.

\bibitem{mackenzie2015}
Andrew Mackenzie.
\newblock Symmetry and impartial lotteries.
\newblock {\em Games and Economic Behavior}, 94:15--28, 2015.

\bibitem{mackenzie2019axiomatic}
Andrew Mackenzie.
\newblock An axiomatic analysis of the papal conclave.
\newblock {\em Economic Theory}, 69:713--743, 2020.

\bibitem{mattei2020peernomination}
Nicholas Mattei, Paolo Turrini, and Stanislav Zhydkov.
\newblock Peer{N}omination: Relaxing exactness for increased accuracy in peer
  selection.
\newblock In {\em Proceedings of the 29th International Joint Conference on
  Artificial Intelligence (IJCAI)}, pages 393--399, 2020.

\bibitem{MR95}
Rajeev Motwani and Prabhakar Raghavan.
\newblock {\em Randomized Algorithms}.
\newblock Cambridge University Press, 1995.

\bibitem{O58}
Masashi Okamoto.
\newblock Some inequalities relating to the partial sum of binomial
  probabilities.
\newblock {\em Annals of the Institute of Statistical Mathematics},
  10(1):29--35, 1958.

\bibitem{tamura2016characterizing}
Shohei Tamura.
\newblock Characterizing minimal impartial rules for awarding prizes.
\newblock {\em Games and Economic Behavior}, 95:41--46, 2016.

\bibitem{tamura2014impartial}
Shohei Tamura and Shinji Ohseto.
\newblock Impartial nomination correspondences.
\newblock {\em Social Choice and Welfare}, 43(1):47--54, 2014.

\end{thebibliography}
\appendix
\section{Appendix: multiplicative approximation and voter correlation }
In the following, we briefly justify two main decisions that we have taken. First, we show that knowing the prior cannot help us improve the approximation ratio of $2$ that is best possible for worst-case inputs. This explains why we have completely ignored the study of multiplicative approximations when prior information is available. 

We extend the  approximation ratio $\rho$ of a mechanism $f$ against a prior $\mathbf{P}$ as follows: 
\begin{align*}
	\rho &= \frac{\E_{\x\sim\mathbf{P}}[\Delta(\x)]}{\E_{\x\sim\mathbf{P}}[d_{f(\x)}(\x)]}
\end{align*}
	
\begin{theorem}
	For every $\epsilon>0$, no impartial selection mechanism has approximation ratio better than $2-\epsilon$ against all uniform priors.
\end{theorem}

\begin{proof}
	Consider uniform instances with two nodes $u$ and $v$ of popularity $p$. Clearly, $\E[\Delta(\x)]=1-(1-p)^2=2p-p^2$. We show that for every impartial mechanism $f$, it holds $\E[d_{f(\x)}(\x)]\leq p$. The theorem then follows by taking $p$ to be sufficiently small.
	
	Indeed, consider the profile consisting of the two directed edges between $u$ and $v$ and let $q_u$ and $q_v$ be the probabilities that the winner is node $u$ and node $v$, respectively. Impartiality means that node $u$ is the winner with probability $q_u$ at the profile consisting only of the directed edge from $v$ to $u$ and node $v$ is the winner at the profile consisting only of the directed edge from $u$ to $v$ with probability $q_v$. Overall, 
	\begin{align*}
	\E[d_{f(\x)}(\x)] &=(q_u+q_v) \cdot p^2 + q_u\cdot p \cdot (1-p) + q_v \cdot (1-p)\cdot p \leq (q_u+q_v)\cdot p \leq p.
	\end{align*}
	Notice that our argument includes randomized mechanisms that may return no winner with positive probability at some profiles.
\end{proof}

Second, we show that our assumption about voter independence is crucial since, otherwise, even our most appealing AVD mechanism has linear additive approximation.

\begin{example}\label{ex:1}
	Consider the following instance with $8k+2$ nodes partitioned into sets of nodes $A$ and $B$ of $4k$ nodes each and two additional nodes $a$ and $b$. Node $a$ is approved by no node with probability $1/2$ and all the $4k$ nodes of set $A$ with probability $1/2$ (i.e., there is correlation between the votes in $A$). Similarly, and independently from the approvals to node $a$, node $b$ is approved by no node with probability $1/2$ and by all nodes of set $B$ with probability $1/2$. Notice that there is always a tie and hence AVD always selects the default node, which cannot have expected in-degree higher than $2k$. The expected highest in-degree is $4k$ with probability $3/4$ and $0$ with probability $1/4$, i.e., an expected highest in-degree of $3k$. Hence, the additive approximation is $k$, i.e., linear in the number of nodes.
\end{example}

\end{document}